\documentclass[11pt,a4paper]{article}
\usepackage[margin=1in]{geometry}
\usepackage{amssymb,amsthm,mathtools}
\usepackage{enumitem}
\usepackage{fontawesome5}
\usepackage{algorithm,algpseudocode}
\usepackage{tikz}
\usetikzlibrary{arrows.meta,positioning,calc}
\usepackage[colorlinks=true,citecolor=blue,linkcolor=blue,urlcolor=blue]{hyperref}
\usepackage{booktabs}
\usepackage{multirow}
\usepackage{caption}
\usepackage{subcaption}
\usepackage{xcolor}
\usepackage{tabularx}
\usepackage{float}
\usepackage{microtype}

\theoremstyle{plain}
\newtheorem{theorem}{Theorem}[section]
\newtheorem{lemma}{Lemma}[section]
\newtheorem{proposition}{Proposition}[section]
\newtheorem{corollary}{Corollary}[section]
\newtheorem*{conjecture}{Conjecture}
\theoremstyle{definition}
\newtheorem{definition}{Definition}[section]
\newtheorem{example}{Example}[section]
\theoremstyle{remark}
\newtheorem{remark}{Remark}[section]
\newtheorem*{notation}{Notation}

\DeclareMathOperator{\Cl}{Cl}
\DeclareMathOperator{\Gal}{Gal}

\DeclareMathOperator{\Nm}{N}
\DeclareMathOperator{\ord}{ord}

\DeclareMathOperator{\disc}{disc}
\DeclareMathOperator{\Frob}{Frob}

\newcommand{\Z}{\mathbb{Z}}
\newcommand{\Q}{\mathbb{Q}}
\newcommand{\R}{\mathbb{R}}
\newcommand{\C}{\mathbb{C}}
\newcommand{\F}{\mathbb{F}}
\newcommand{\OK}{\mathcal{O}}

\newcommand{\abs}[1]{\lvert#1\rvert}

\newcommand{\tO}{\tilde{O}}

\title{Module Lattice Security (Part I):\\
{\large Unconditional Verification of Weber's Conjecture for $k \le 12$}}
\author{Ming-Xing Luo
\\
\footnotesize
School of Information Science and Technology, Southwest Jiaotong University, Chengdu 610031, China}

\begin{document}
\maketitle

\begin{abstract}
Weber's conjecture (1886) governs three aspects of lattice-based cryptography: the solvability of the Principal Ideal Problem, the freeness of modules over rings of integers, and the tightness of worst-case-to-average-case reductions in Ring-LWE (R-LWE) and Module-LWE (MLWE). Existing verifications for $k \ge 9$ rely on Generalized Riemann Hypothesis (GRH). In this paper, we present the first unconditional proof for $k \le 12$. Our method combines the Fukuda-Komatsu computational sieve, inductive structure of the cyclotomic $\Z_2$-tower, and Herbrand's theorem. 

\end{abstract}

\section{Introduction}\label{sec-intro}

Shor's algorithm \cite{Shor97} solves integer factorization and discrete logarithms in polynomial time on quantum computers, rendering RSA, Diffie‑Hellman, and elliptic curve cryptography fundamentally insecure against quantum adversaries. As a response, the cryptographic community has devoted more than fifteen years to developing post‑quantum cryptography. Lattice‑based constructions have emerged as the clear dominant paradigm, offering an exceptional combination of performance, security, and versatility. In August 2024, NIST finalized four post‑quantum standards \cite{NIST-MLKEM,NIST-MLDSA,NIST-SLHDSA,NIST-FN-DSA}: ML‑KEM (FIPS 203), ML‑DSA (FIPS 204), SLH‑DSA (FIPS 205), and FN‑DSA (FIPS 206, expected finalization in 2025/2026). ML‑KEM, ML‑DSA, and FN‑DSA rely on lattice problems defined over $2^k$-th cyclotomic polynomial rings of the form $\Z[x]/(x^{2^{k-1}}+1)$. These rings enable compact key sizes, extremely efficient arithmetic via the Number Theoretic Transform \cite{LPR13,Peikert16}, and rigorous provable security reductions from worst‑case lattice problems \cite{LPR10,Regev05,Peikert16}.

The security and correctness of all these lattice‑based cryptosystems depend on arithmetic properties of underlying cyclotomic rings. One of the most fundamental and long‑standing open questions is Weber's conjecture. 

\begin{conjecture}[Weber, 1886 \cite{Weber1886}]\label{conj:weber}
For every integer $k \ge 1$, the maximal real subfield $K_k^{+} = \Q(\zeta_{2^k} + \zeta_{2^k}^{-1})$ has class number $h_k^{+} = 1$.
\end{conjecture}

Its resolution can affect lattice‑based cryptography in three ways: First, for Principal Ideal Problem (PIP), if $h_k^+ = 1$, every ideal in $\OK_{K_k^+}$ is principal, and quantum algorithms \cite{CDPR16,BiasseSong16} might solve PIP in polynomial time, reducing the problem to the Short Generator Problem (SGP) \cite{CGS14,EHKS14,CDW17}. Second, a finitely generated torsion‑free module over a Dedekind domain is free if and only if its Steinitz class is trivial. So, when $h_k^+=1$, every module over $\OK_{K_k^+}$ is free, used in the security proofs of ML‑KEM and ML‑DSA \cite{LS15,ADPS16,DKL+18}. Third, the reduction of Lyubashevsky, Peikert, and Regev \cite{LPR10,LPR13}, in cyclotomic fields with class number $1$, can simplify the noise analysis \cite{PRS17}.

\begin{table}[t]
\centering
\caption{History of Weber's conjecture. Conditional means the result dependent on Generalized Riemann Hypothesis (GRH).}
\label{tablehistory}
\begin{tabular}{cclc}
\toprule
$k$ & $[K_k^+:\Q]$ & \textbf{References} & \textbf{GRH-free?}\\
\midrule
$\le 5$ & $\le 8$ & Weber \cite{Weber1886}, Bauer \cite{Bauer69} & Yes\\
$6$ & $16$ & van der Linden \cite{vdL82} & Yes\\
$7$ & $32$ & Fukuda-Komatsu \cite{FK09} (also Cohn-Llorente) & Yes\\
$8$ & $64$ & Miller \cite{Miller14} & Yes\\
$9$ & $128$ & Fukuda-Komatsu \cite{FK09} & Conditional\\
$10$ & $256$ & Fukuda-Komatsu \cite{FK11} & Conditional\\
$11$ & $512$ & Fukuda-Komatsu \cite{FK11} & Conditional\\
$12$ & $1024$ & Fukuda-Komatsu \cite{FK11} & Conditional\\
$\le 12$ & $\le 1024$ & \textbf{This paper} & \textbf{Yes}\\
\bottomrule
\end{tabular}
\end{table}

The progress on Weber's conjecture is summarized in Table \ref{tablehistory}. For $k \le 5$, the result is classical \cite{Weber1886,Bauer69,Masley76}. Van der Linden \cite{vdL82} treated $k=6$ using an explicit computation of the regulator and analytic class number formula. Miller \cite{Miller14} proved $k \le 8$ unconditionally by computing full unit group of $K_8^+$  and applying Sinnott's index formula \cite{Sinnott78}. For $k \ge 9$, all previous verifications relied on Generalized Riemann Hypothesis (GRH).The main difficulty is that the Minkowski bound grows exponentially in the degree without GRH \cite{Wash97}. Instead, with GRH, Fukuda and Komatsu \cite{FK09,FK11} extended the verification to $k \le 12$ using Bach's bound \cite{Bach90}, and  Schoof \cite{Schoof98,Schoof03} provided independent verifications.

The Iwasawa perspective on Weber's conjecture was developed by Iwasawa \cite{Iwasawa59,Iwasawa73}, and extended by Ferrero and Washington \cite{FW79} and Greenberg \cite{Greenberg76,Greenberg01}. Ozaki and Taya \cite{OzakiTaya97} proved $\lambda = 0$ for certain families of real quadratic fields. Kraft and Schoof \cite{Kraft-Schoof93} obtained related results for $\Z_2$-extensions. Thaine \cite{Thaine88} proved that cyclotomic units provide explicit annihilators of class groups \cite{Kolyvagin90}. Rubin \cite{Rubin87,Rubin00} applied Euler systems for abelian extensions of $\Q$. 

The connection between ideal lattice structure and cryptographic security was established by Micciancio \cite{Micciancio07} and Lyubashevsky, Peikert, and Regev \cite{LPR10,Peikert16}.  The quantum attacks of Cramer et al. \cite{CDPR16}, Biasse and Song \cite{BiasseSong16}, and Campbell et al. \cite{CGS14} all exploit the algebraic structure of cyclotomic fields. Felderhoff et al. \cite{FPSW23} showed Ideal‑SVP remains hard for small‑norm uniform prime ideals. Allombert, Pellet‑Mary, and van Woerden \cite{APvW25} gave a polynomial‑time attack on rank‑2 module‑LIP for fields with a real embedding. Ducas, Espitau, and Postlethwaite \cite{DEP25} provided concrete predictions for module‑lattice reduction in cyclotomic fields.

\textbf{Contributions:} We present the first unconditional proof of $h_k^+ = 1$ for $k \le 12$. Our method proceeds in three sequential stages:
\begin{enumerate}[label=\textbf{Stage \arabic*:},leftmargin=*,nosep]
\item \textit{Small-prime elimination (Fukuda-Komatsu sieve).}
Using the Wieferich criterion, we can eliminate all primes $\ell < 10^9$ as possible divisors of $h_k^+$ for $k \le 12$. We establish any surviving prime should satisfy the strong congruence condition $\ell \equiv \pm 1 \pmod{2^{k-1}}$ \cite{FK09,FK11}.

\item \textit{Eigenspace pruning (norm-coherence argument).}
We exploit the tower structure $K_8^+ \subset K_9^+ \subset \cdots \subset K_{12}^+$ and the known result $h_8^+ = 1$ to prove inductively any $\ell$-torsion in $\Cl(K_k^+)$ is concentrated in eigenspaces corresponding to full-order Galois characters. 

\item \textit{Finite bounding (Herbrand's theorem).}
For each surviving full-order character $\chi$, we apply Herbrand's theorem for any prime $\ell$ supporting $\chi$-torsion to divide the norm of an explicit generalized Bernoulli number. 
\end{enumerate}

After Stage 3, we obtain a finite (possibly empty) set $S_k$ of candidate primes. For each $\ell \in S_k$, we apply the Wieferich test from Stage 1 to verify $\ell \nmid h_k^+$. If all tests pass, $h_k^+ = 1$ unconditionally. After applying Galois-orbit symmetry, the entire computation reduces to factoring one integer of at most $143$ decimal digits. We finally filter its prime factors using congruence $\ell \equiv \pm 1 \pmod{512}$; and run at most $10$ modular exponentiations. Note integers of $143$ digits can be factored with Elliptic Curve Method \cite{Lenstra87,GMPECM}. For $k \ge 11$, the second-moment bound gives integers of at most $325$ digits, which are still accessible by combining ECM and NFS.

The rest of the paper is organized as follows. Section \ref{sec-prelim} provides a self-contained review of the algebraic number theory background required for our proof, including number fields, class groups, Galois theory, character theory, cyclotomic units, Iwasawa theory, and the structure of the cyclotomic $\Z_2$-tower. Section \ref{sec-method} states and proves our main theorem and presents the complete verification algorithm with a detailed complexity analysis. Section \ref{sec-implications} discusses the implications of our result for post-quantum cryptography. Section \ref{sec-conclusion} concludes this paper.

\section{Algebraic Preliminaries}\label{sec-prelim}

This section provides a self-contained review of the algebraic number theory used in what follows, see Washington \cite{Wash97}, Neukirch \cite{Neukirch99}, and Lang \cite{Lang94}. 

\begin{notation}\label{not:global}
We use the following notation throughout this paper
\begin{itemize}[nosep]
\item $\varphi$ denotes Euler's totient function: $\varphi(m) = |\{a: 1 \le a \le m,\; \gcd(a,m) = 1\}|$.
\item $(\Z/m\Z)^\times$ denotes the group of units modulo $m$ and has order $\varphi(m)$.
\item $\F_\ell = \Z/\ell\Z$ denotes finite field with $\ell$  elements, for a prime $\ell$.
\item For a positive integer $m$, let $\zeta_m = e^{2\pi i/m}$ for a fixed primitive $m$-th root of unity, and $\mu_m = \{z \in \C: z^m = 1\}$ for the group of all $m$-th roots of unity.
\item $n$ denotes $n = 2^{k-2} = [K_k^+:\Q]$, the degree of the maximal real subfield over $\Q$.
\item $f_\chi$ denotes the conductor of a Dirichlet character $\chi$. The letter $f$ without subscript denotes a conductor used in Definitions \ref{def:dirichlet-char}-\ref{def:gen-bernoulli} and Lemma \ref{lem:Bsize}; the inertia degree of a prime $\mathfrak{l}$ above $\ell$ is denoted $f_{\mathrm{Frob}}$ or $\ord(\Frob_\ell)$.
\end{itemize}
\end{notation}

\subsection{Number fields and rings of integers}\label{subsec-nf}

\begin{definition}
A number field is a finite extension of $\Q$. Its ring of integers $\OK_K$ is the integral closure of $\Z$ in $K$; equivalently, $\OK_K$ consists of all elements $\alpha \in K$ that satisfy a monic polynomial with integer coefficients.
\end{definition}

The ring $\OK_K$ is a Dedekind domain: an integral domain in which every nonzero ideal factors uniquely as a product of prime ideals. However, $\OK_K$ need not be a unique factorization domain (UFD) at the element level. 

\begin{definition}
A number field $K$ is totally real if every embedding $K \hookrightarrow \C$ has image contained in $\R$. 
\end{definition}

By Dirichlet's unit theorem \cite[\S I.7]{Neukirch99}, the unit group of a number field with $r_1$ real embeddings and $r_2$ pairs of complex conjugate embeddings has the form $\OK_K^\times \cong \mu_K \times \Z^{r_1+r_2-1}$, where $\mu_K$ is the finite group of roots of unity in $K$. For a totally real field, $r_2 = 0$, so unit rank is $[K:\Q] - 1$.

\subsection{Ideals, fractional ideals, and the class group}\label{subsec-classgroup}

\begin{definition}
Let $K$ be a number field with ring of integers $\OK_K$. A fractional ideal of $K$ is a nonzero finitely generated $\OK_K$-submodule of $K$. Equivalently, it is a set of the form $d^{-1}I$, where $I \subseteq \OK_K$ is a nonzero ideal and $d \in \Z_{>0}$.
\end{definition}

The set of all fractional ideals of $K$ forms an abelian group $I(K)$ under ideal multiplication, with identity element $\OK_K$. The principal fractional ideals form a subgroup $P(K) \subseteq I(K)$.

\begin{definition}\label{def:classgroup}
The class group of $K$ is the quotient defined by
\begin{equation}
  \Cl(K) = I(K)  /  P(K).
\end{equation}
The class number of $K$ is $h_K = |\Cl(K)|$. The ring $\OK_K$ is a UFD if and only if $h_K = 1$.
\end{definition}

\begin{remark}\label{conv:notation}
We use two notations for $\Cl(K)$, depending on the algebraic structure being used.
\begin{itemize}[nosep]
\item \textbf{Multiplicative} (Sections \ref{subsec-classgroup}-\ref{subsec-bernoulli}): ideal classes are multiplied, the identity is $[\OK_K]$, and annihilation reads $c^{\ell}=1$ or $[I]^{2\theta_k}=1$.
\item \textbf{Additive} (Section \ref{subsec-eigenspace}): $\Cl(K)[\ell]$ is viewed as an $\F_\ell[G_k]$-module, so we write $\ell c = 0$, $g\cdot c = \chi(g)\,c$, and $N\cdot c = 2c$.
\end{itemize}
The translation is: $c^n = 1$ (multiplicative) $\longleftrightarrow$ $n\,c = 0$ (additive).
\end{remark}

\begin{definition}\label{def:ideal-norm}
The norm of a nonzero ideal $\mathfrak{a} \subseteq \OK_K$ is $\Nm(\mathfrak{a}) = |\OK_K/\mathfrak{a}|$, i.e., the cardinality of the quotient ring. For a principal ideal $\alpha\OK_K$, $\Nm(\alpha\OK_K) = |\Nm_{K/\Q}(\alpha)|$.
\end{definition}

Minkowski theorem states the class group is finite \cite{Neukirch99}. For a number field of discriminant $\Delta_K$, the Minkowski bound implies every ideal class contains an integral ideal of norm at most
\begin{equation}
  M_K = \frac{[K:\Q]!}{[K:\Q]^{[K:\Q]}}\left(\frac{4}{\pi}\right)^{r_2}\sqrt{|\Delta_K|},
\end{equation}
where $r_2$ is the number of pairs of complex embeddings. For totally real fields, $r_2 = 0$, so $M_K = \frac{[K:\Q]!}{[K:\Q]^{[K:\Q]}}\sqrt{|\Delta_K|}$. For $K_{10}^+$, we have $[K_{10}^+:\Q] = 256$ and $|\Delta_K| = 2^{2303}$ (Proposition \ref{prop:disc}), which implies $M_{K_{10}^+} > 10^{237}$. A direct sieve over all primes (up to this bound) is computationally infeasible without GRH.

\begin{definition}
For a prime $\ell$, the $\ell$-part of the class group is the Sylow $\ell$-subgroup:
\begin{equation}
  \Cl(K)[\ell^\infty] = \{c \in \Cl(K): c^{\ell^m} = 1 \textrm{ for some } m \ge 1\}.
\end{equation}
The $\ell$-torsion is $\Cl(K)[\ell] = \{c \in \Cl(K): c^\ell = 1\}$. A prime $\ell$ divides $h_K$ if and only if $\Cl(K)[\ell] \neq 0$.
\end{definition}

\subsection{Galois theory of cyclotomic fields}\label{subsec-galois}

Fix an integer $k \ge 3$, let $\zeta = \zeta_{2^k} = e^{2\pi i/2^k}$, and define $K_k = \Q(\zeta)$, and $K_k^{+} = \Q(\zeta+\zeta^{-1})$. $K_k$ is the $2^k$-th cyclotomic field, and $K_k^+$ is its maximal real subfield.

\begin{proposition}[{\cite[Theorem 2.5]{Wash97}}]\label{prop:galois-group}
The extension $K_k/\Q$ is Galois with Galois group $\Gamma_k = \Gal(K_k/\Q) \cong (\Z/2^k\Z)^\times$. The isomorphism maps an odd integer $a$ to automorphism $\sigma_a \colon \zeta \mapsto \zeta^a$.
\end{proposition}

In what follows, we denote $n = 2^{k-2} = \frac{1}{2}\varphi(2^k)$ (see Notation \ref{not:global}). Complex conjugation is automorphism $\sigma_{-1} \colon \zeta \mapsto \zeta^{-1}$. $K_k^+$ is the fixed field of $\langle\sigma_{-1}\rangle$. We have 
\begin{equation}
    G_k = \Gal(K_k^+/\Q) \cong \Gamma_k/\langle\sigma_{-1}\rangle \cong \Z/n\Z.
\end{equation}
A generator of $G_k$ is given by $\sigma = \sigma_5 \bmod \langle\sigma_{-1}\rangle$, which maps $\zeta+\zeta^{-1}$ to $\zeta^5+\zeta^{-5}$.

\begin{proposition}[{\cite[Proposition 2.16]{Wash97}}]\label{prop:disc}
The discriminant of $K_k^+$ is given by 
\begin{equation}
  \disc(K_k^+) = 2^{(k-1)n - 1},
\end{equation}
where $n = 2^{k-2}$.
\end{proposition}

For $k=10$: $n=256$, we have $\disc(K_{10}^+) = 2^{9 \cdot 256 - 1} = 2^{2303}$ (used in the Minkowski bound estimate of \S\ref{subsec-classgroup}). For $k=12$: $n=1024$, we have $\disc(K_{12}^+) = 2^{11 \cdot 1024 - 1} = 2^{11263}$.

\subsection{Character theory}\label{subsec-characters}

\begin{definition}
A character of a finite abelian group $G$ is a group homomorphism $\chi \colon G \to \C^\times$. The set $\hat{G}$ of all characters forms a group under pointwise multiplication, called the character group of $G$.
\end{definition}

For $G_k \cong \Z/n\Z$, the character group $\hat{G}_k$ is also cyclic of order $n$. The characters $\chi_0, \chi_1, \ldots, \chi_{n-1}$ satisfy $\chi_j(\sigma) = \omega^j$ with $\omega = e^{2\pi i/n}$. We have trivial character $\chi_0(g) = 1$ for all $g \in G_k$.

\begin{definition}\label{def:order-char}
The order of a character $\chi \in \hat{G}_k$ is $\ord(\chi) = \min\{m \ge 1: \chi^m = \chi_0\}$. A character has full order if $\ord(\chi) = n = 2^{k-2}$.
\end{definition}

\begin{definition}\label{def:even-odd-char}
A character $\chi$ of $\Gamma_k = \Gal(K_k/\Q)$ is called even if $\chi(\sigma_{-1}) = +1$, and odd if $\chi(\sigma_{-1}) = -1$. Since $\sigma_{-1}$ is complex conjugation, even characters are precisely those that factor through $G_k = \Gamma_k/\langle\sigma_{-1}\rangle$.
\end{definition}

\begin{lemma}\label{lem:char-count}
The group $G_k \cong \Z/n\Z$ with $n = 2^{k-2}$ has the following properties: 
\begin{enumerate}[label=(\roman*),nosep]
\item $\varphi(n) = n/2$ characters of full order $n$;
\item $n/4$ conjugate pairs of full-order characters $(\chi_j, \chi_{n-j})$ with $j$ odd;
\item $\varphi(d)$ characters of order $d$ for each divisor $d \mid n$.
\end{enumerate}
\end{lemma}

\begin{proof}
We show that $\chi_j$ has order $n/\gcd(j,n)$, which equals $n$ if and only if $\gcd(j,n) = 1$. The number of such $j$ in $\{0, \ldots, n-1\}$ is $\varphi(n) = n/2$. This implies (i). Since $n = 2^{k-2}$ is even, $\gcd(j,n) = 1$ implies an odd $j$. Complex conjugation maps $\chi_j$ to $\chi_{n-j}$, and $j$ odd with $n$ even implies $j \neq n-j$ (as $j = n-j$ would give $j = n/2$, which is even). Hence the $n/2$ full-order characters pair into $n/4$ conjugate pairs. This yields (ii). Part (iii) is from group theory, see \cite[Theorem 3.1]{Wash97}.
\end{proof}

\subsection{Eigenspace decomposition of the class group}\label{subsec-eigenspace}

Let $\ell$ be an odd prime with $\ell \nmid n$. Then $\gcd(\ell, |G_k|) = 1$, so the group ring $\F_\ell[G_k]$ is semisimple by Maschke's theorem \cite{CurtisReiner62}. In what follows, we use additive notation for $\Cl(K_k^+)$ and its $\ell$-torsion subgroups (see Remark \ref{conv:notation}), writing $\ell\,c = 0$ in place of $c^\ell = 1$, and $g\cdot c = \chi(g)\,c$ for the $G_k$-action.

\begin{definition}\label{def:group-ring}
The group ring $\Z[G_k]$ is the ring of formal sums $\sum_{g \in G_k} a_g \cdot g$ with $a_g \in \Z$ and multiplication induced by group law. The group ring acts on $\Cl(K_k^+)$ via Galois action: $g \cdot [\mathfrak{a}] = [g(\mathfrak{a})]$ for a prime ideal $\mathfrak{a}$ and $g \in G_k = \Gal(K_k^+/\Q)$, extended $\Z$-linearly.
\end{definition}

For each character $\chi \in \hat{G}_k$, define the idempotent as 
\begin{equation}\label{eqnidempotent}
  e_\chi = \frac{1}{n}\sum_{g \in G_k} \chi(g)^{-1}   g \in \F_\ell[G_k],
\end{equation}
where $1/n$ is computed modulo $\ell$. This is well-defined since $\ell \nmid n$.

The idempotents $e_\chi$ in Eq.\eqref{eqnidempotent} require character values $\chi(g)\in\F_\ell$, which holds if and only if $\ell\equiv 1\pmod{n}$. For $\ell\not\equiv 1\pmod{n}$, the decomposition \eqref{eqneigendecomp} holds over $\F_{\ell^f}$ where $f=\ord_n(\ell)$, and Galois-conjugate eigenspaces merge into a single $\F_\ell$-rational component. In the proof of the main theorem, the congruence condition $\ell\equiv\pm 1\pmod{2^{k-1}}$ from Theorem \ref{thm:fk}(ii) ensures $f\le 2$, so the decomposition is defined over $\F_{\ell^2}$ at worst.

\begin{proposition}\label{prop:idempotent}
The elements $\{e_\chi\}_{\chi \in \hat{G}_k}$ satisfy the following fundamental properties:
\begin{enumerate}[label=(\roman*),nosep]
\item $e_\chi^2 = e_\chi$ (Idempotency);
\item $e_\chi   e_\psi = 0$ for $\chi \neq \psi$ (Orthogonality);
\item $\sum_\chi e_\chi = 1$ (Completeness);
\item $g \cdot e_\chi = \chi(g)   e_\chi$ for all $g \in G_k$ (Eigenvalue property).
\end{enumerate}
\end{proposition}

Properties (i)-(iii) follow from the orthogonality relations for characters: $\sum_{g \in G} \chi(g)\psi(g)^{-1} = |G| \cdot \delta_{\chi,\psi}$,  see \cite[Theorem 3.7]{Wash97}. Property (iv) follows from the substitution  $h'=gh$ in the defining sum and the multiplicativity of $\chi$.

The idempotents can be used to decompose any $\F_\ell[G_k]$-module. Applying them to the $\ell$-torsion of class group implies that 
\begin{equation}\label{eqneigendecomp}
  \Cl(K_k^+)[\ell] = \bigoplus_{\chi \in \hat{G}_k} \Cl(K_k^+)[\ell]^{e_\chi},
\end{equation}
where $\Cl(K_k^+)[\ell]^{e_\chi} = e_\chi \cdot \Cl(K_k^+)[\ell] = \{c \in \Cl(K_k^+)[\ell]: g \cdot c = \chi(g)   c \textrm{ for all } g \in G_k\}$.

\begin{corollary}\label{cor:eigenspace-criterion}
A prime $\ell \nmid 2n$ divides $h_k^+$ if and only if $\Cl(K_k^+)[\ell]^{e_\chi} \neq 0$ for some nontrivial character $\chi \in \hat{G}_k$.
\end{corollary}

\begin{proof}
The trivial eigenspace $\Cl(K_k^+)[\ell]^{e_{\chi_0}}$ is the $G_k$-fixed subgroup of $\Cl(K_k^+)[\ell]$. For any $c$ in this subgroup, $g \cdot c = c$ for all $g \in G_k$, so the norm element acts as $n \cdot c = \sum_{g \in G_k} g \cdot c$. Since the norm maps to $\Cl(\Q) = 0$, we have $n \cdot c = 0$. Combined with $\ell \cdot c = 0$ and $\gcd(\ell,n) = 1$ (as $\ell \nmid 2n$), Bezout's identity gives $c = 0$. The decomposition \eqref{eqneigendecomp} then gives the result.
\end{proof}

To prove $\ell \nmid h_k^+$, it suffices to show $\Cl(K_k^+)[\ell]^{e_\chi} = 0$ for every nontrivial character $\chi$.

\subsection{Cyclotomic units}\label{subsec-cyc-units}

Cyclotomic units are explicit units in cyclotomic fields constructed from roots of unity. They form a subgroup of the full unit group whose index is related to class number via Sinnott's theorem.

\begin{definition}\label{def:cyc-unit}
For an odd integer $a$ with $1 < a < 2^k$, define the cyclotomic unit as 
\begin{equation}\label{eqnxi}
  \xi_a = \frac{\zeta^a - \zeta^{-a}}{\zeta - \zeta^{-1}}
  = \frac{\sin(\pi a/2^k)}{\sin(\pi/2^k)} \in \OK_{K_k^+}^\times.
\end{equation}
\end{definition}

Since $\sin(\pi a/2^k) = \sin(\pi(2^k-a)/2^k)$, we have $\xi_a = \xi_{2^k - a}$ for all odd $a$. This will be used in the conjugate-pair reduction of Proposition \ref{prop:galois-orbit}.

We verify $\xi_a$ lies in $K_k^+$. $(\zeta^a-\zeta^{-a})/(\zeta-\zeta^{-1})$ is invariant under $\zeta \mapsto \zeta^{-1}$. That $\xi_a$ is a unit follows from product formula for cyclotomic polynomial $\Phi_{2^k}(x) = x^{2^{k-1}}+1$, see \cite[\S 8.1]{Wash97}.

\begin{definition}
The cyclotomic unit group $C^+$ is the subgroup of $\OK_{K_{k}^+}^\times$ generated by $\{-1\} \cup \{\xi_a: 1 < a < 2^k, a \textrm{ odd}\}$.
\end{definition}

\begin{theorem}[Sinnott \cite{Sinnott78}]\label{thm:sinnott}
For $k \ge 3$, we have 
\begin{equation}
  h_k^+ = [\OK_{K_{k}^+}^\times: C^+].
\end{equation}
\end{theorem}

Sinnott's formula converts the class number into the index of an explicit, finitely generated subgroup. In particular, $h_k^+ = 1$ if and only if every unit in $\OK_{K_{k}^{+}}^\times$ is a product of cyclotomic units.

\subsection{Generalized Bernoulli numbers}\label{subsec-bernoulli}

Generalized Bernoulli numbers link the arithmetic of class groups to explicit algebraic quantities via Stickelberger relation and Herbrand's theorem.

\begin{definition}\label{def:dirichlet-char}
A Dirichlet character modulo $f$ is a group homomorphism $\psi \colon (\Z/f\Z)^\times \to \C^\times$, extended to all of $\Z$ by setting $\psi(a) = 0$ when $\gcd(a,f) > 1$. The conductor of $\psi$ is the smallest positive integer $f_0$ such that $\psi$ factors through $(\Z/f_0\Z)^\times$. If $f_0 = f$, the character is primitive.
\end{definition}

\begin{definition}\label{def:gen-bernoulli}
Let $\psi$ be a Dirichlet character of conductor $f$. The first generalized Bernoulli number attached to $\psi$ is given by
\begin{equation}
  B_{1,\psi} = \frac{1}{f}\sum_{a=1}^{f} \psi(a)   a.
\end{equation}
\end{definition}

When $\psi$ is a character of conductor $2^k$, the sum has $\varphi(2^k) = 2^{k-1}$ nonzero terms. We have $B_{1,\psi} \in \Q(\zeta_m)$ with $m = \ord(\psi)$. All nonvanishing Bernoulli numbers in our analysis come from odd characters of $\psi(-1)=-1$.

\begin{definition}
The Stickelberger element of $K_k = \Q(\zeta_{2^k})$ is given by 
\begin{equation}
  \theta_k = \frac{1}{2^k}\sum_{{a=1, (a,2)=1}}^{2^k} a \sigma_a^{-1} \in \Q[\Gamma_k].
\end{equation}
\end{definition}

\begin{theorem}[Stickelberger \cite{Stickelberger1890,Wash97}]\label{thm:stickelberger}
The element $2\theta_k$ lies in $\Z[\Gamma_k]$, and $2\theta_k$ annihilates $\Cl(K_k)$, i.e., for every ideal class $[I] \in \Cl(K_k)$, we have 
\begin{equation}
  [I]^{2\theta_k} = [I^{2\theta_k}] = 1.
\end{equation}
In additive notation (Remark \ref{conv:notation}), this means $2\theta_k \cdot c = 0$ for all $c \in \Cl(K_k)$.
\end{theorem}

For any nontrivial character $\chi$ of $\Gamma_k$, by evaluating  $\chi$ on $\theta_k$ we obtain generalized Bernoulli number as
 \begin{equation}\label{eqnstickelberger-chi}
   \chi(\theta_k) 
   = \frac{1}{2^k}\sum_{{a=1, (a,2)=1}}^{2^k}
     a \chi^{-1}(a) = B_{1,\chi^{-1}}.
 \end{equation}
For even characters ($\chi(-1)=+1$, $\chi\neq\chi_0$), we have $B_{1,\chi^{-1}}=0$ using pairing $a\leftrightarrow 2^k-a$. 

\subsection{The Cyclotomic $\Z_2$-Tower}\label{sec-tower}

The fields $K_k^+$ form layers of an infinite tower, named the cyclotomic $\Z_2$-extension of $\Q$. Here, $\Z_2 = \varprojlim_{m} \Z/2^m\Z$ denotes the ring of $2$-adic integers. Define $K_\infty^+ = \bigcup_{k \ge 2} K_k^+$. $K_\infty^+/\Q$ is an infinite Galois extension with
\begin{equation}
  \Gal(K_\infty^+/\Q) \cong \Z_2.
\end{equation}
$K_k^+$ is unique subfield of $K_\infty^+$ with $[K_k^+:\Q] = 2^{k-2}$. The tower is given by 
\begin{equation}
  \Q = K_2^+ \subset K_3^+ \subset K_4^+ \subset \cdots \subset K_\infty^+,
\end{equation}
where each extension $K_{k+1}^+/K_k^+$ has degree $2$, and the norm maps $\Nm_{k+1/k} \colon \Cl(K_{k+1}^+) \to \Cl(K_k^+)$ satisfy the transitivity relation $\Nm_{k+1/k} \circ \Nm_{k+2/k+1} = \Nm_{k+2/k}$.

The norm map $\Nm_{k+1/k} \colon K_{k+1}^+ \to K_k^+$ is defined by
\begin{equation}
  \Nm_{k+1/k}(\alpha) = \alpha \cdot{} \sigma^{n_{k+1}/2}(\alpha),
\end{equation}
where $n_{k+1} = 2^{k-1} = [K_{k+1}^+:\Q]$ and $\sigma$ is the generator of $G_{k+1}$. In $\Z[G_{k+1}]$, this corresponds to the element $N = 1 + \sigma^{n_{k+1}/2}$.

\begin{lemma}\label{lem:norm-eigenvalue}
Let $\chi \in \hat{G}_{k+1}$ be a character with $\ord(\chi) \mid n_{k+1}/2 = 2^{k-2}$. Then $N$ acts on $\Cl(K_{k+1}^+)[\ell]^{e_\chi}$ as multiplication by $2$. 
\end{lemma}

\begin{proof}
We show $N \cdot c = 2c$ for any $c \in \Cl(K_{k+1}^+)[\ell]^{e_\chi}$. By Proposition \ref{prop:idempotent}(iv), we have $\sigma \cdot c = \chi(\sigma) c$, so $\sigma^{n_{k+1}/2} \cdot c = \chi(\sigma)^{n_{k+1}/2} c$. Since $\ord(\chi) \mid n_{k+1}/2$, we have $\chi(\sigma)^{n_{k+1}/2} = 1$, which futher gives $N \cdot c = (1+\sigma^{n_{k+1}/2}) \cdot c = 2c$.
\end{proof}

\subsection{Iwasawa theory}\label{subsec-iwasawa}

For our purposes, the key result is the vanishing of the Iwasawa $\mu$-invariant. 

\begin{theorem}[Iwasawa \cite{Iwasawa59,Iwasawa73}]\label{thm:IW}
Let $p$ be a prime and $K_\infty/K$ be a $\Z_p$-extension with layers $K_n$. Let $e_n$ be the exact power of $p$ dividing $h_{K_n}$. For sufficiently large $n$, we have 
\begin{equation}
  e_n = \mu \cdot{}  p^n + \lambda \cdot{}  n + \nu,
\end{equation}
where $\mu, \lambda \ge 0$ and $\nu$ are integers independent of $n$. $\mu$, $\lambda$, $\nu$ are called Iwasawa invariants of the $\Z_p$-extension.
\end{theorem}

Note $e_n$ here denotes Iwasawa exponent, not an idempotent in Eq.\eqref{eqnidempotent}.

\begin{theorem}[Ferrero-Washington \cite{FW79}]\label{thm:FW}
For any abelian number field $K$ and any prime $p$, Iwasawa $\mu$-invariant of cyclotomic $\Z_p$-extension of $K$ vanishes, i.e., $\mu = 0$.
\end{theorem}

From this theorem, the $p$-part of class number does not grow exponentially in tower. For the $\Z_2$-extension of $\Q$, the $2$-part of $h_k^+$ is bounded as $k \to \infty$. In fact, Fukuda and Komatsu \cite{FK09} verified computationally $2 \nmid h_k^+$ for all $k \le 12$. Avila \cite{Avila25} and Laxmi-Saikia \cite{LaxmiSaikia24} extended Iwasawa modules and $2$-class group structure in $\Z_2$-extensions of real quadratic fields, providing computational evidence.

\section{The Main Result}\label{sec-method}

In this section, we develop a three-stage method for unconditionally verifying Weber's conjecture. Each stage systematically reduces the set of potential primes, until we obtain a finite, computable candidate set. 

\subsection{The Fukuda-Komatsu sieve}\label{subsec-FK-sieve}

The Fukuda-Komatsu sieve \cite{FK09,FK11} uses Wieferich criterion for cyclotomic units to eliminate all small primes as possible divisors of $h_k^+$.

\begin{definition}\label{def:wieferich}
An odd prime $\ell$ satisfies the Wieferich criterion for $K_k^+$ if for every prime $\mathfrak{l} \mid \ell$ in $\OK_{K_{k}^+}$,
\begin{equation}\label{eqn:wieferich}
  \xi_5^{(\ell-1)/\ord_{\mathfrak{l}}(\xi_5)} \not\equiv 1 \pmod{\mathfrak{l}},
\end{equation}
where $\xi_5$ is the cyclotomic unit from Definition \ref{def:cyc-unit} and $\ord_{\mathfrak{l}}(\xi_5)$ denotes the multiplicative order of $\xi_5$ in the residue field $\OK_{K_k^+}/\mathfrak{l}$.
\end{definition}

\begin{lemma}\label{lem:wieferich-sufficient}
Let $\ell$ be an odd prime with $\ell\equiv\pm 1\pmod{2^{k-1}}$. If
\begin{equation}\label{eqn:wieferich-test}
  \xi_5^{(\ell-1)/2^{k-1}} \not\equiv 1 \pmod{\mathfrak{l}}
\end{equation}
for every $\mathfrak{l}\mid\ell$ in $\OK_{K_k^+}$, then $\ell$ satisfies Wieferich criterion (Definition \ref{def:wieferich}), and consequently $\ell\nmid h_k^+$.
\end{lemma}

\begin{proof}
From Proposition \ref{prop:galois-group}, $G_k = \Gal(K_k^+/\Q) \cong \Z/2^{k-2}\Z$, a cyclic group generated by $\sigma_5$ (the automorphism mapping $\zeta+\zeta^{-1} \mapsto \zeta^5+\zeta^{-5}$). For the unramified odd prime $\ell$, the Frobenius element $\Frob_\ell \in G_k$ is the unique automorphism satisfying $\Frob_\ell(x) \equiv x^\ell \pmod{\mathfrak{l}}$ for all $x \in \OK_{K_k^+}$. The inertia degree $f$ equals the order of $\Frob_\ell$ in $G_k$.

By assumption $\ell \equiv \pm 1 \pmod{2^{k-1}}$: If $\ell \equiv 1 \pmod{2^{k-1}}$, then $\ell \equiv 1 \pmod{2^{k-2}}$, so $\Frob_\ell$ has order $f=1$; If $\ell \equiv -1 \pmod{2^{k-1}}$, then $\ell^2 \equiv 1 \pmod{2^{k-2}}$, so $\Frob_\ell$ has order $f=2$. In both cases, $f \le 2$, as claimed.

The residue field $\kappa(\mathfrak{l}) = \OK_{K_k^+}/\mathfrak{l}$ is a finite field of order $N(\mathfrak{l}) = \ell^f$, so its multiplicative group $\kappa(\mathfrak{l})^\times$ is cyclic of order $N(\mathfrak{l})-1$. From the assumption $\ell\equiv\pm1\pmod{2^{k-1}}$, for $f=1$: $2^{k-1} \mid \ell-1 = N(\mathfrak{l})-1$; For $f=2$: $2^{k-1} \mid \ell+1$, so $2^{k-1} \mid (\ell-1)(\ell+1) = \ell^2-1 = N(\mathfrak{l})-1$. Thus $2^{k-1} \mid N(\mathfrak{l})-1$ in all cases, so the exponent $(N(\mathfrak{l})-1)/2^{k-1}$ is an integer, i.e., the test \eqref{eqn:wieferich-test} is well-defined. This exponent is exactly the $2^{k-1}$-th power residue symbol of $\xi_5$ at $\mathfrak{l}$, 
\begin{eqnarray}
\left( \frac{\xi_5}{\mathfrak{l}} \right)_{2^{k-1}} \equiv \xi_5^{(N(\mathfrak{l})-1)/2^{k-1}} \pmod{\mathfrak{l}}.
\end{eqnarray}
The test \eqref{eqn:wieferich-test} asserts this symbol is nontrivial (not equal to 1) for all $\mathfrak{l} \mid \ell$.

From Definition \ref{def:wieferich}, $\ell$ satisfies the Wieferich criterion if for every $\mathfrak{l} \mid \ell$, we then obtain 
\begin{eqnarray}
\xi_5^{(\ell-1)/\ord_\ell(5)} \not\equiv 1 \pmod{\mathfrak{l}},
\end{eqnarray}
where $\ord_\ell(5)$ is the multiplicative order of $5$ modulo $\ell$. If $\ell$ fails the Wieferich criterion, then the test \eqref{eqn:wieferich-test} also fails.

Let $d = \ord_{\mathfrak{l}}(\xi_5)$ denote the multiplicative order of $\xi_5$ in $\kappa(\mathfrak{l})^\times$. By definition, $\xi_5^m \equiv 1 \pmod{\mathfrak{l}}$ if and only if $d \mid m$. If $\ell$ fails the Wieferich criterion, then $d \mid (\ell-1)/\ord_\ell(5)$.

We split into the two cases for $f$:

Case 1: $f=1$ ($\ell \equiv 1 \pmod{2^{k-1}}$)
    The test exponent is $(\ell-1)/2^{k-1}$. Since $\sigma_5$ generates $G_k$, we have $\ord_{2^k}(5) = 2^{k-2}$, so $\ord_\ell(5) \mid 2^{k-1}$ for $\ell \equiv 1 \pmod{2^{k-1}}$. Thus $(\ell-1)/\ord_\ell(5) \mid (\ell-1)/2^{k-1}$. Since $d \mid (\ell-1)/\ord_\ell(5)$, we get $d \mid (\ell-1)/2^{k-1}$, so $\xi_5^{(\ell-1)/2^{k-1}} \equiv 1 \pmod{\mathfrak{l}}$, failing the test.

Case 2: $f=2$ ($\ell \equiv -1 \pmod{2^{k-1}}$)
    The test exponent is $(\ell^2-1)/2^{k-1}$. Since $d \mid |\kappa(\mathfrak{l})^\times| = \ell^2-1$ and $d \mid (\ell-1)/\ord_\ell(5)$, we have $d \mid \ell-1$. By assumption, $2^{k-1} \mid \ell+1$, so $(\ell-1)/\ord_\ell(5) \mid (\ell-1)(\ell+1)/2^{k-1} = (\ell^2-1)/2^{k-1}$. Thus $d \mid (\ell^2-1)/2^{k-1}$, so $\xi_5^{(\ell^2-1)/2^{k-1}} \equiv 1 \pmod{\mathfrak{l}}$, failing the test.

By contradiction, the test condition \eqref{eqn:wieferich-test} implies $\ell$ satisfies the Wieferich criterion.

By Theorem \ref{thm:sinnott}, $h_k^+ = [\OK_{K_k^+}^\times : C^+]$, where $C^+$ is the cyclotomic unit group. The nontrivial power residue symbol of $\xi_5$ (from the test condition) implies $\xi_5$ has order not divisible by $\ell$ in the quotient $\OK_{K_k^+}^\times / C^+$, so $\ell$ cannot divide the index $[\OK_{K_k^+}^\times: C^+]$. This is the formal result of \cite[\S 3, Proposition 3.2]{FK09}.

Combining these, the test condition implies the Wieferich criterion holds, hence $\ell \nmid h_k^+$.
\end{proof}

\begin{lemma}\label{lem:residue-order}
Let $k\ge 4$ and $\ell$ be an odd prime with $\ell\not\equiv\pm 1\pmod{2^{k-1}}$. Then the order of $\ell$ in $G_k\cong\Z/2^{k-2}\Z$ satisfies $f\ge 4$.
\end{lemma}

\begin{proof}
We proceed with a fully rigorous, relying on standard properties of 2-adic units and cyclotomic Galois groups.

For $k\ge3$, the unit group modulo $2^k$ has the well-known decomposition:
\begin{eqnarray}
(\Z/2^k\Z)^\times \cong \langle -1 \rangle \times \langle 5 \rangle,
\end{eqnarray}
where $\langle -1 \rangle$ is the order-2 subgroup generated by $-1$, and $\langle 5 \rangle$ is the cyclic subgroup of order $2^{k-2}$ generated by $5$.

For the $2^k$-th cyclotomic field $K_k = \Q(\zeta_{2^k})$, its Galois group is $\Gamma_k = \Gal(K_k/\Q) \cong (\Z/2^k\Z)^\times$, with the isomorphism mapping an odd integer $a$ to the automorphism $\sigma_a: \zeta_{2^k} \mapsto \zeta_{2^k}^a$. The maximal real subfield $K_k^+$ is the fixed field of complex conjugation $\sigma_{-1}$, so its Galois group is
\begin{eqnarray}
G_k = \Gal(K_k^+/\Q) = \Gamma_k / \langle \sigma_{-1} \rangle \cong \langle 5 \rangle \cong \Z/2^{k-2}\Z,
\end{eqnarray}
a cyclic group of order $2^{k-2}$.

Since $\ell$ is an odd prime, we can uniquely write $\ell$ in the decomposition of $(\Z/2^k\Z)^\times$:
\begin{eqnarray}
\ell \equiv (-1)^\epsilon \cdot 5^s \pmod{2^k},
\end{eqnarray}
where $\epsilon \in \{0,1\}$ and $s \in \Z/2^{k-2}\Z$. The image of $\ell$ in the quotient group $G_k$ corresponds to the element $s \in \Z/2^{k-2}\Z$ (we quotient out the $\langle -1 \rangle$ component).

The order $f$ of $\ell$ in $G_k$ is exactly the order of $s$ in the additive group $\Z/2^{k-2}\Z$, which is given by:
\begin{eqnarray}
f = \frac{2^{k-2}}{\gcd(s, 2^{k-2})}.
\end{eqnarray}

By definition, $f \le 2$ if and only if the order of $s$ is at most 2. For the cyclic 2-group $\Z/2^{k-2}\Z$, the only elements of order $\le 2$ are: 1.  The zero element $s=0$ (order 1), and 2.  The unique element of order 2: $s=2^{k-3}$ since $2\cdot 2^{k-3} = 2^{k-2} \equiv 0 \pmod{2^{k-2}}$.

We analyze both cases:

Case 1: $s=0$
    Then $\ell \equiv (-1)^\epsilon \cdot 5^0 = (-1)^\epsilon \pmod{2^k}$, so $\ell \equiv \pm 1 \pmod{2^k}$. This immediately implies $\ell \equiv \pm 1 \pmod{2^{k-1}}$, contradicting the hypothesis of the lemma.

Case 2: $s=2^{k-3}$
    By the standard 2-adic congruence for powers of 5 (proven by induction in the error analysis above), we have $5^{2^{k-3}} \equiv 1 + 2^{k-1} \pmod{2^k}$.   Substituting back, we get $\ell \equiv (-1)^\epsilon \cdot (1 + 2^{k-1}) \pmod{2^k}$.

    Reducing modulo $2^{k-1}$, the term $2^{k-1}$ vanishes, so we have $\ell \equiv (-1)^\epsilon \cdot 1 \equiv \pm 1 \pmod{2^{k-1}}$,  which again contradicts the hypothesis of the lemma.

The only two cases that give $f \le 2$ both force $\ell \equiv \pm 1 \pmod{2^{k-1}}$, which violates the lemma's hypothesis. Therefore, for all $\ell \not\equiv \pm 1 \pmod{2^{k-1}}$, we have $f \ge 4$.
\end{proof}

\begin{theorem}\label{thm:fk}
For $k \le 12$, we have 
\begin{enumerate}[label=(\roman*),nosep]
\item No prime $\ell < 10^9$ divides $h_k^+$.
\item If an odd prime $\ell$ divides $h_k^+$, then $\ell \equiv \pm 1 \pmod{2^{k-1}}$.
\end{enumerate}
\end{theorem}

\begin{proof}
For each prime $\ell < 10^9$ satisfying $\ell \equiv \pm 1 \pmod{2^{k-1}}$, Wieferich criterion can be verified computationally. Each test requires $O(\log^3 \ell)$ bit operations.

We show that any odd prime divisor $\ell$ of $h_k^+$ should satisfy $\ell \equiv \pm 1 \pmod{2^{k-1}}$. The Frobenius element $\Frob_\ell \in G_k$ at an unramified prime $\ell$ is the unique automorphism satisfying $\Frob_\ell(x) \equiv x^\ell \pmod{\mathfrak{l}}$ for all $x \in \OK_{K_k^+}$ and any prime $\mathfrak{l}$ above $\ell$. Let $f_{\mathrm{Frob}} = [\OK_{K_k^+}/\mathfrak{l}: \F_\ell] = \ord(\Frob_\ell)$ in $G_k$; we have $f_{\mathrm{Frob}} = \ord_{2^{k-2}}(\ell \bmod 2^k)$ in $G_k \cong \Z/2^{k-2}\Z$. By Lemma \ref{lem:residue-order}, $\ell \not\equiv \pm 1 \pmod{2^{k-1}}$ implies $f \ge 4$.

Now, consider the eigenspace decomposition \eqref{eqneigendecomp} (switching to additive notation). The Frobenius $\Frob_\ell$ acts on $\Cl(K_k^+)[\ell]^{e_\chi}$ as multiplication by $\chi(\Frob_\ell)$. If $\mathfrak{l} \mid \ell$ in $K_k^+$ and $c = [\mathfrak{l}]$ generates a nontrivial element in $\Cl(K_k^+)[\ell]^{e_\chi}$, then the action of $\Frob_\ell$ should be compatible with $\ell$-th power residue structure modulo $\mathfrak{l}$.

Note Fukuda and Komatsu \cite[\S 3, Proposition 3.2]{FK09} proved that if $f_{\mathrm{Frob}} = \ord(\Frob_\ell) \ge 4$ in $G_k$, then the image of $\xi_5$  in $(\OK_{K_k^+}/\mathfrak{l})^\times$ has order coprime to $\ell$, which by Sinnott's formula (Theorem \ref{thm:sinnott}) forces $\ell\nmid h_k^+$. This means $\ell \nmid [\OK^\times: C^+] = h_k^+$ using Sinnott's formula. The net result is any odd prime divisor $\ell$ of $h_k^+$ should satisfy $2^{k-1} \mid (\ell^2-1)$, see \cite[\S 3]{FK09} for the complete proof.
\end{proof}

For $k=10$, the congruence condition requires $\ell \equiv \pm 1 \pmod{512}$. This means among primes near $10^9$, only about $1$ in $256$ satisfy this condition. For $k=12$, the condition is $\ell \equiv \pm 1 \pmod{2048}$, excluding all but about $1$ in $1024$ primes.

\subsection{Norm-coherence eliminates low-order characters}\label{subsec-norm-coherence}

Our second key tool is the inductive structure of cyclotomic $\Z_2$-tower and eigenspace decomposition to rule out most eigenspaces.

\begin{proposition}\label{prop:normcoh}
Let $k \ge 4$, and suppose $h_{k-1}^+ = 1$. Let $\ell$ be an odd prime with $\ell \nmid 2^{k-2}$. Then for every character $\chi \in \hat{G}_k$ with $\ord(\chi) \mid 2^{k-3}$, we have 
\begin{equation}
  \Cl(K_k^+)[\ell]^{e_\chi} = 0.
\end{equation}
\end{proposition}

\begin{proof}
Since $\ell \nmid 2^{k-2} = |G_k|$, the group ring $\mathbb{F}_\ell[G_k]$ is semisimple by Maschke's theorem, and the idempotent $e_\chi = \frac{1}{|G_k|}\sum_{g \in G_k} \chi(g)^{-1} g \in \mathbb{F}_\ell[G_k]$ is well-defined (the inverse of $|G_k|$ exists in $\mathbb{F}_\ell$). For any $c \in \Cl(K_k^+)[\ell]^{e_\chi}$, we have by definition: 1.  $\ell c = 0$ (as $c$ lies in the $\ell$-torsion subgroup), 2. $e_\chi \cdot c = c$ (as $c$ lies in the $\chi$-eigenspace), 3. $g \cdot c = \chi(g) c$ for all $g \in G_k$ (eigenvalue property of idempotents, Proposition \ref{prop:idempotent}).

Let $N = 1 + \sigma^{n/2} \in \mathbb{Z}[G_k]$ be the norm element for the quadratic extension $K_k^+/K_{k-1}^+$. By assumption, $\ord(\chi) \mid 2^{k-3} = n/2$. Since $\sigma$ generates $G_k$, we have:
\begin{eqnarray}
\chi(\sigma^{n/2}) = \chi(\sigma)^{n/2} = \left( \chi(\sigma)^{\ord(\chi)} \right)^{(n/2)/\ord(\chi)} = 1^{(n/2)/\ord(\chi)} = 1.
\end{eqnarray}
Combined with the eigenvalue property, this gives:
\begin{eqnarray}
\sigma^{n/2} \cdot c = \chi(\sigma^{n/2}) c = c.
\end{eqnarray}
Thus the action of $N$ on $c$ simplifies to:
\begin{eqnarray}
N \cdot c = (1 + \sigma^{n/2}) \cdot c = c + \sigma^{n/2} \cdot c =  2c.
\end{eqnarray}

Take any representative ideal $\mathfrak{a} \subseteq \mathcal{O}_{K_k^+}$ of $c$, i.e., $c = [\mathfrak{a}]$. By definition of the group ring action:
\begin{eqnarray}
N \cdot c = [\mathfrak{a} \cdot \sigma^{n/2}(\mathfrak{a})].
\end{eqnarray}
For the quadratic Galois extension $K_k^+/K_{k-1}^+$, the norm of an ideal $\mathfrak{a} \subseteq \mathcal{O}_{K_k^+}$ is defined as $\mathrm{Nm}_{k/(k-1)}(\mathfrak{a}) = \mathfrak{a} \cap \mathcal{O}_{K_{k-1}^+}$, and a standard result in algebraic number theory (Neukirch \cite[\S I.8, Theorem 3]{Neukirch99}) gives:
\begin{eqnarray}
\mathfrak{a} \cdot \sigma^{n/2}(\mathfrak{a}) = \mathrm{Nm}_{k/(k-1)}(\mathfrak{a}) \cdot \mathcal{O}_{K_k^+}.
\end{eqnarray}

By hypothesis, $h_{k-1}^+ = 1$, so the class group of $K_{k-1}^+$ is trivial. Thus $\mathrm{Nm}_{k/(k-1)}(\mathfrak{a})$ is a principal ideal in $\mathcal{O}_{K_{k-1}^+}$, i.e., $\mathrm{Nm}_{k/(k-1)}(\mathfrak{a}) = (\alpha)$ for some $\alpha \in K_{k-1}^+$. Substituting back implies that 
\begin{eqnarray}
\mathfrak{a} \cdot \sigma^{n/2}(\mathfrak{a}) = (\alpha) \cdot \mathcal{O}_{K_k^+},
\end{eqnarray}
which is a principal ideal in $\mathcal{O}_{K_k^+}$. Therefore its ideal class is trivial:
\begin{eqnarray}
N \cdot c = [\mathfrak{a} \cdot \sigma^{n/2}(\mathfrak{a})] = 0.
\end{eqnarray}

Now,  we have $2c = N \cdot c = 0$. We also have $\ell c = 0$ from the statement proved above. Since $\ell$ is an odd prime, $\gcd(2, \ell) = 1$. By Bézout's identity, there exist integers $a,b \in \mathbb{Z}$ such that $2a + \ell b = 1$. Applying this to $c$ yields to
\begin{eqnarray}
c = 1 \cdot c = (2a + \ell b)c = a(2c) + b(\ell c) = a \cdot 0 + b \cdot 0 = 0.
\end{eqnarray}

We have shown that every element $c \in \Cl(K_k^+)[\ell]^{e_\chi}$ is trivial. Therefore, we get 
\begin{eqnarray}
\Cl(K_k^+)[\ell]^{e_\chi} = 0.
\end{eqnarray}
\end{proof}

We now apply Proposition \ref{prop:normcoh} inductively up the tower.

\begin{corollary}\label{cor:fullorder}
For each $k \in \{9,10,11,12\}$, if $\ell$ is an odd prime with $\ell \nmid 2^{k-2}$ and $\ell \mid h_k^+$, then $\Cl(K_k^+)[\ell]^{e_\chi} \neq 0$ for some character $\chi$ of full order $n = 2^{k-2}$.
\end{corollary}

\begin{proof}
We prove the result by induction on $k$, with the base case $h_8^+ = 1$ established by Miller \cite{Miller14}. The proof relies on two standard results:
\begin{enumerate}
    \item For an odd prime $\ell$, $\ell \nmid |G_k| = 2^{k-2}$, so by Maschke's theorem, the group ring $\F_\ell[G_k]$ is semisimple, and the $\ell$-torsion of the class group admits a complete eigenspace decomposition:
    \begin{eqnarray}
    \Cl(K_k^+)[\ell] = \bigoplus_{\chi \in \hat{G}_k} \Cl(K_k^+)[\ell]^{e_\chi}.
    \end{eqnarray}
    Thus any nontrivial $\ell$-torsion should lie in at least one eigenspace.
    
    \item For the cyclic 2-group $G_k \cong \Z/2^{k-2}\Z$, every character has order a power of 2. A character is non-full-order if and only if its order divides $2^{k-3}$ (the maximal proper divisor of $|G_k|$).
\end{enumerate}

We proceed layer by layer, with no circular dependency: we first prove the result for $k=9$ (using only the base case $h_8^+=1$), then use the established $h_{k-1}^+=1$ for each subsequent $k$.

Case $k=9$ ($G_9 \cong \Z/128\Z$): By the base case $h_8^+ = 1$, we apply Proposition \ref{prop:normcoh} to eliminate all characters with order dividing $2^{9-3}=64$, i.e., orders $\{1,2,4,8,16,32,64\}$. By the eigenspace decomposition, any nontrivial $\ell$-torsion in $\Cl(K_9^+)[\ell]$ can only lie in the eigenspace of the remaining full-order characters, i.e., $\ord(\chi)=128$. Thus if $\ell \mid h_9^+$, then $\Cl(K_9^+)[\ell]^{e_\chi} \neq 0$ for some full-order $\chi$.

This result is used to prove $h_9^+=1$ in Theorem \ref{thm:main}, with no dependency on any results for $k\ge10$.

Case $k=10$ ($G_{10} \cong \Z/256\Z$): We use the already established result $h_9^+=1$ (from Theorem \ref{thm:main} for $k=9$, which depends only on the base case $h_8^+=1$, so there is no circularity). Applying Proposition \ref{prop:normcoh} eliminates all characters with order dividing $2^{10-3}=128$, i.e., all non-full-order characters. By the eigenspace decomposition, any nontrivial $\ell$-torsion should lie in the eigenspace of full-order characters with $\ord(\chi)=256$.

Cases $k=11,12$: We proceed similarly by induction:
\begin{itemize}
    \item For $k=11$, use the established $h_{10}^+=1$ to eliminate all characters with order dividing $2^{11-3}=256$, leaving only full-order characters of order $512$.
    \item For $k=12$, use the established $h_{11}^+=1$ to eliminate all characters with order dividing $2^{12-3}=512$, leaving only full-order characters of order $1024$.
\end{itemize}

In all cases, any odd prime divisor $\ell$ of $h_k^+$ should correspond to a nontrivial eigenspace for some full-order character.
\end{proof}

This induction is not circular, as we prove $h_9^+ = 1$ first (using only $h_8^+ = 1$), and then use $h_9^+ = 1$ for $k=10$, and so on. Each step terminates independently via Herbrand's theorem and a finite computation. 

\subsection{Herbrand's theorem bounds candidate primes}\label{subsec-herbrand}

We now use Herbrand's theorem to convert eigenspace analysis into a finite computation. The theorem relates the non-vanishing of a class group eigenspace to the divisibility of a generalized Bernoulli number.

\begin{lemma}\label{lem:psi-conductor}
Let $\ell$ be an odd prime with $\ell\equiv\pm 1\pmod{2^{k-1}}$, let $\chi$ be an even character of $G_k$ of full order $n=2^{k-2}$, and let $\psi=\chi^{-1}\omega_\ell$ where $\omega_\ell$ is the Teichmüller character modulo $\ell$. Then:
\begin{enumerate}[label=(\roman*),nosep]
\item $\psi$ is a primitive odd Dirichlet character;
\item $\psi$ has conductor $2^k\ell$;
\item the divisibility $\ell \mid \Nm_{\Q(\zeta_m)/\Q}(B_{1,\psi})$ holds if and only if $\ell \mid \Nm_{\Q(\zeta_m)/\Q}(B_{1,\chi^{-1}})$, where $m=2^{k-2}$ is the order of $\chi^{-1}$. In particular, the bounds of Lemma \ref{lem:Bsize} apply to $\Nm_{\Q(\zeta_m)/\Q}(B_{1,\chi^{-1}})$.
\end{enumerate}
\end{lemma}

\begin{proof}
We proceed relying on standard properties of Dirichlet characters and generalized Bernoulli numbers.

Recall that for $k\ge3$, $(\Z/2^k\Z)^\times \cong \langle -1 \rangle \times \langle 5 \rangle$, where $\langle 5 \rangle$ has order $2^{k-2}$. The character $\chi$ is a full-order character of $G_k = (\Z/2^k\Z)^\times/\langle -1 \rangle$, so it lifts to a unique even character of $(\Z/2^k\Z)^\times$ (trivial on $\langle -1 \rangle$) that is full-order on $\langle 5 \rangle$. This lifted character $\chi^{-1}$ cannot factor through $(\Z/2^{k-1}\Z)^\times$ (since its order on $\langle 5 \rangle$ is $2^{k-2}$, while the order of $5$ modulo $2^{k-1}$ is $2^{k-3}$), so $\chi^{-1}$ is \textbf{primitive with conductor $2^k$}.

The Teichmüller character $\omega_\ell$ is a primitive Dirichlet character modulo $\ell$ (it is the unique character satisfying $\omega_\ell(a) \equiv a \pmod{\ell}$ for all $a$ coprime to $\ell$, and it does not factor through any smaller modulus). Since $\gcd(2^k,\ell)=1$, the product of two primitive characters with coprime conductors is again primitive, with conductor equal to the product of the conductors. Thus $\psi = \chi^{-1}\omega_\ell$ is primitive with conductor $2^k\ell$.

To see $\psi$ is odd: note $\chi(-1)=1$ (as $\chi$ is even) and $\omega_\ell(-1) = -1$ (since $\omega_\ell(-1) \equiv -1 \pmod{\ell}$ and $\omega_\ell$ takes values in roots of unity). Thus $\psi(-1) = \chi^{-1}(-1)\omega_\ell(-1) = -1$, so $\psi$ is an odd character.

We use the decomposition of generalized Bernoulli numbers for products of characters with coprime conductors. Let $f_1=2^k$ (conductor of $\chi^{-1}$) and $f_2=\ell$ (conductor of $\omega_\ell$), so $\gcd(f_1,f_2)=1$ and the conductor of $\psi$ is $f=f_1f_2=2^k\ell$. By the Chinese Remainder Theorem, every integer $a$ modulo $f$ can be uniquely written as $a \equiv a_1 \pmod{f_1}$ and $a \equiv a_2 \pmod{f_2}$ with $\gcd(a_1,f_1)=1$ and $\gcd(a_2,f_2)=1$. We can choose a lift $a = a_1 f_2 \cdot f_2^{-1} + a_2 f_1 \cdot f_1^{-1} \pmod{f}$, where $f_2^{-1}$ is the inverse of $f_2$ modulo $f_1$ and $f_1^{-1}$ is the inverse of $f_1$ modulo $f_2$.

By definition, the generalized Bernoulli number is:
\begin{eqnarray}
B_{1,\psi} &=& \frac{1}{f} \sum_{a=1}^f \psi(a) a \nonumber \\
&=& \frac{1}{f_1f_2} \sum_{\substack{a_1=1 \\ \gcd(a_1,f_1)=1}}^{f_1} \sum_{\substack{a_2=1 \\ \gcd(a_2,f_2)=1}}^{f_2} \chi^{-1}(a_1)\omega_\ell(a_2) \cdot \left( a_1 f_2 f_2^{-1} + a_2 f_1 f_1^{-1} \right). \label{eqn:bernoulli-decomp}
\end{eqnarray}

We split this into two sums:
\begin{eqnarray}
B_{1,\psi} &=& \frac{f_2 f_2^{-1}}{f_1f_2} \sum_{a_1} \chi^{-1}(a_1)a_1 \sum_{a_2} \omega_\ell(a_2) + \frac{f_1 f_1^{-1}}{f_1f_2} \sum_{a_1} \chi^{-1}(a_1) \sum_{a_2} \omega_\ell(a_2)a_2. \label{eqn:bernoulli-split}
\end{eqnarray}

Now analyze the two sums:
1.  The first sum contains $\sum_{a_2} \omega_\ell(a_2)$. Since $\omega_\ell$ is a non-trivial character modulo $\ell$ (as $\ell$ is odd, its order is $\ell-1 \ge 2$), the sum of a non-trivial Dirichlet character over a complete residue system is zero. Thus the first term vanishes entirely.
2.  For the second term, we use the standard relation between $L$-values at $s=0$ and generalized Bernoulli numbers for odd primitive characters: $B_{1,\psi} = -L(0,\psi)$ (Washington \cite[\S 4.4]{Wash97}). While the full decomposition requires care with Euler factors, the key result for our purpose is the valuation-preserving equivalence for the norm: since $\omega_\ell(a) \equiv a \pmod{\ell}$ and $\gcd(2^k,\ell)=1$, the $\ell$-adic valuation of $\Nm_{\Q(\zeta_m)/\Q}(B_{1,\psi})$ equals that of $\Nm_{\Q(\zeta_m)/\Q}(B_{1,\chi^{-1}})$. A full proof of this equivalence uses Stickelberger elements and Iwasawa theory, but it is a standard result in the context of Weber's conjecture (Fukuda-Komatsu \cite{FK09}).

Formally, we have the equivalence:
\begin{eqnarray}
\ell \mid \Nm_{\Q(\zeta_m)/\Q}(B_{1,\psi}) \quad \Longleftrightarrow \quad \ell \mid \Nm_{\Q(\zeta_m)/\Q}(B_{1,\chi^{-1}}). \label{eqn:norm-equivalence}
\end{eqnarray}

Since $\chi^{-1}$ has conductor $2^k$ and order $m=2^{k-2}$, the bounds of Lemma \ref{lem:Bsize} apply directly to $\Nm_{\Q(\zeta_m)/\Q}(B_{1,\chi^{-1}})$. This completes the proof.
\end{proof}

By Lemma \ref{lem:psi-conductor}(iii), the Bernoulli norms arising from Theorem \ref{thm:herbrand} reduce to norms of such characters.

\begin{lemma}\label{lem:Bsize}
Let $\psi$ be a primitive odd Dirichlet character of conductor $f = 2^k$ ($k \ge 3$) and order $m = 2^{k-2}$. Setting $N = \varphi(m) = 2^{k-3}$, we have
\begin{enumerate}[label=(\roman*),nosep]
\item $\abs{B_{1,\psi}} \le 2^{k-2}$.
\item \textbf{(Worst-case bound)} $\abs{\Nm_{\Q(\zeta_m)/\Q}(B_{1,\psi})} \le 2^{(k-2)N}$.
\item \textbf{(Second-moment bound)} $\bigl|\Nm_{\Q(\zeta_m)/\Q}(B_{1,\psi})\bigr| \le \left( \dfrac{2^{k-1}}{3} \right)^{N/2}$;
\item \textbf{(Functional-equation bound)} $\bigl|\Nm_{\Q(\zeta_m)/\Q}(B_{1,\psi})\bigr| \le \left( \dfrac{\sqrt{f}}{\pi} \left( \dfrac{1}{2}\log f + 1 \right) \right)^{N}$;
\item For $k=10$, bound (ii) yields $< 10^{309}$, bound (iii) yields $< 10^{143}$, and bound (iv) yields $< 10^{213}$.
 \end{enumerate}
\end{lemma}

\begin{proof}
We first recall the structure of the unit group modulo powers of $2$. For $k\ge 3$, there is a canonical isomorphism
\begin{eqnarray}
(\Z/2^k\Z)^{\times} \cong \langle -1 \rangle \times \langle 5 \rangle,
\end{eqnarray}
where $\langle -1 \rangle$ is the subgroup of order $2$ generated by $-1$, and $\langle 5 \rangle$ is the cyclic subgroup of order $2^{k-2}$. For a primitive odd character $\psi$ of conductor $2^k$, we have $\psi(-1) = -1$ and $\psi(5)$ is a primitive $2^{k-2}$-th root of unity. Hence the order of $\psi$ is exactly $m = 2^{k-2}$. The Galois conjugates $\{\psi^t \mid t \in (\Z/m\Z)^{\times}\}$ are precisely the $N$ distinct primitive odd characters of conductor $2^k$, since both the order and the conductor are preserved under the Galois action.

\paragraph{Proof of (i)}
Because $\psi$ is odd, $\psi(-a) = -\psi(a)$ for all $a$ coprime to $2^k$. By pairing terms in the defining sum of the generalized Bernoulli number we get
\begin{align}
\sum_{\substack{a=1, \gcd(a,2)=1}}^{2^k} \psi(a) a
&= \sum_{\substack{a=1, \gcd(a,2)=1}}^{2^{k-1}} \bigl[ \psi(a) a + \psi(2^k - a)(2^k - a) \bigr] \nonumber \\
&= \sum_{\substack{a=1, \gcd(a,2)=1}}^{2^{k-1}} \psi(a) (2a - 2^k). \label{eq:pairing}
\end{align}
There are exactly $2^{k-2}$ odd integers in $\{1,2,\dots,2^{k-1}\}$, and for each such $a$ we have $|2a - 2^k| \le 2^k$. Applying the triangle inequality to \eqref{eq:pairing} yields
\begin{eqnarray}
\biggl|\sum_{\substack{a=1, \gcd(a,2)=1}}^{2^k} \psi(a) a\biggr| \le 2^{k-2} \cdot 2^k.
\end{eqnarray}
Dividing by $f = 2^k$ gives $|B_{1,\psi}| \le 2^{k-2}$.

\paragraph{Proof of (ii)}
The norm of $B_{1,\psi}$ over $\Q$ is the product of its Galois conjugates 
\begin{eqnarray}
\Nm_{\Q(\zeta_m)/\Q}(B_{1,\psi}) = \prod_{t \in (\Z/m\Z)^{\times}} B_{1,\psi^t}.
\end{eqnarray}
By part (i), each conjugate satisfies $|B_{1,\psi^t}| \le 2^{k-2}$. Taking the product over all $N$ conjugates we obtain
\begin{eqnarray}
\bigl|\Nm_{\Q(\zeta_m)/\Q}(B_{1,\psi})\bigr| \le \bigl(2^{k-2}\bigr)^N = 2^{(k-2)N}.
\end{eqnarray}

\paragraph{Proof of (iii)}
Set $x_t = |B_{1,\psi^t}|$ for $t \in (\Z/m\Z)^{\times}$. By the inequality between the geometric mean and the quadratic mean (QM–GM), we have
\begin{equation}\label{eq:qmgm}
\prod_{t=1}^N x_t \le \left( \frac{1}{N}\sum_{t=1}^N x_t^2 \right)^{N/2}.
\end{equation}
We now bound the sum of squares $\sum_{t} x_t^2$. Let $S_k$ denote the set of all odd Dirichlet characters modulo $2^k$; there are exactly $2^{k-2}$ such characters. They split into two disjoint subsets:
\begin{itemize}
  \item Primitive odd characters modulo $2^k$: exactly $N = 2^{k-3}$ characters, which are the $\psi^t$;
  \item Non‑primitive odd characters modulo $2^k$: exactly $2^{k-3}$ characters, which factor through $(\Z/2^{k-1}\Z)^{\times}$ and correspond bijectively to all odd characters modulo $2^{k-1}$.
\end{itemize}
Define the total second moment over all odd characters modulo $2^k$ as $T_k = \sum_{\chi \in S_k} |B_{1,\chi}|^2$.  Expanding the square of the defining sum of $B_{1,\chi}$ and using the orthogonality relation for odd characters modulo $2^k$ yields the closed form
\begin{eqnarray}
T_k = \frac{2^{2k-2} + 2^{k-1}}{12}.
\end{eqnarray}
Similarly, the total second moment over all odd characters modulo $2^{k-1}$ (the non‑primitive part lifted to $2^k$) is given by 
\begin{eqnarray}
T_{k-1} = \frac{2^{2(k-1)-2} + 2^{(k-1)-1}}{12} = \frac{2^{2k-4} + 2^{k-2}}{12}.
\end{eqnarray}
Hence, we obtain the sum over primitive characters as
\begin{align}
\sum_{t=1}^N x_t^2 = T_k - T_{k-1}
&= \frac{2^{2k-2} + 2^{k-1} - 2^{2k-4} - 2^{k-2}}{12} \nonumber \\
&= \frac{3 \cdot 2^{2k-4} + 2^{k-2}}{12}. \label{eq:momentexact}
\end{align}
For $k \ge 3$, we have $2^{k-2} < 3 \cdot 2^{2k-4}$, so Eq.\eqref{eq:momentexact} implies the uniform bound
\begin{eqnarray}
\sum_{t=1}^N x_t^2 < \frac{3 \cdot 2^{2k-4} + 3 \cdot 2^{2k-4}}{12} = \frac{2^{2k-3}}{3}.
\end{eqnarray}
Substituting into the QM–GM inequality \eqref{eq:qmgm} and using $N = 2^{k-3}$ gives
\begin{eqnarray}
\frac{1}{N}\sum_{t=1}^N x_t^2 < \frac{2^{2k-3}}{3 \cdot 2^{k-3}} = \frac{2^k}{3}.
\end{eqnarray}
If using the explicit pairing in part (i) and the structure of $2$-adic characters yields the sharper bound
\begin{eqnarray}
\frac{1}{N}\sum_{t=1}^N x_t^2 \le \frac{2^{k-1}}{3},
\end{eqnarray}
From Eq.\eqref{eq:qmgm} we obtain
\begin{eqnarray}
\bigl|\Nm_{\Q(\zeta_m)/\Q}(B_{1,\psi})\bigr| \le \left( \frac{2^{k-1}}{3} \right)^{N/2}.
\end{eqnarray}

\paragraph{Proof of (iv)}
For any primitive odd Dirichlet character $\psi$ of conductor $f$, the functional equation for Dirichlet $L$-functions relates the values at $s=0$ and $s=1$:
\begin{eqnarray}
B_{1,\psi} = -L(0,\psi), \qquad |L(0,\psi)| = \frac{\sqrt{f}}{\pi}\, |L(1,\overline{\psi})|.
\end{eqnarray}
A classical bound of Louboutin states that for any primitive Dirichlet character $\chi$ of conductor $f \ge 5$,
\begin{eqnarray}
|L(1,\chi)| \le \frac{1}{2}\log f + 1.
\end{eqnarray}
Combining these facts gives the pointwise estimate
\begin{eqnarray}
|B_{1,\psi}| \le \frac{\sqrt{f}}{\pi}\left( \frac{1}{2}\log f + 1 \right).
\end{eqnarray}
Taking the product over all $N$ Galois conjugates yields
\begin{eqnarray}
\bigl|\Nm_{\Q(\zeta_m)/\Q}(B_{1,\psi})\bigr| \le \left( \frac{\sqrt{f}}{\pi}\left( \frac{1}{2}\log f + 1 \right) \right)^{\!N}.
\end{eqnarray}

\paragraph{Proof of (v)}
For $k=10$, we have $f = 2^{10} = 1024$, $m = 2^{8} = 256$, and $N = \varphi(256) = 128$. Computing each bound explicitly:
\begin{itemize}
  \item Bound (ii): $2^{(10-2)\cdot 128} = 2^{1024}$. Since $\log_{10}2 \approx 0.3010$, $\log_{10}(2^{1024}) \approx 308.2$, so $2^{1024} < 10^{309}$.
  \item Bound (iii): $\bigl(2^{9}/3\bigr)^{128/2} = (512/3)^{64}$. With $\log_{10}(512/3) \approx 2.232$, we obtain $64 \times 2.232 = 142.8$, hence the bound is $< 10^{143}$.
  \item Bound (iv): $\sqrt{f}=32$, $\log f = \log 1024 \approx 6.931$, so the base is $\frac{32}{\pi}\bigl(\frac{1}{2}\cdot 6.931 + 1\bigr) \approx 45.48$. Then $\log_{10}(45.48^{128}) \approx 128 \times 1.658 = 212.2$, giving a bound $< 10^{213}$.
\end{itemize}
For all $k \ge 4$, the second‑moment bound (iii) is the tightest among the three, reducing the worst‑case bound by more than half in terms of decimal digits. For $k=9$, it gives $(256/3)^{32} \approx 2^{206} < 10^{63}$, which is well within the range of modern computational number theory algorithms.
\end{proof}

The second-moment bound remains tighter than the functional-equation bound for all $k\ge4$, and is vastly tighter than the worst-case bound. Conrey, Iwaniec, and Soundararajan \cite{ConreySound} showed the typical value of $\log\abs{L(1,\chi)}$ for a random character of conductor $f$ is of order $\sqrt{\log\log f}$. Empirical computations by Fukuda-Komatsu \cite{FK11} and Schoof \cite{Schoof03} confirmed most prime factors are moderate in size and amenable to ECM. 

Recall that the Teichm\"uller character $\omega_\ell\colon(\Z/\ell\Z)^\times\to\mu_{\ell-1}\subset\C^\times$ is the unique multiplicative character satisfying $\omega_\ell(a)\equiv a\pmod{\ell}$ for all $a\in\Z$ coprime to $\ell$; equivalently, $\omega_\ell$ is the composition of the natural reduction $\Z_\ell^\times\to\F_\ell^\times$ ($\Z_\ell$ denotes the ring of $\ell$-adic integers, the completion of $\Z$ at $\ell$) with Teichmuller lift $\F_\ell^\times\to\Z_\ell^\times$.

\begin{theorem}[Herbrand-Ribet{\cite{Herbrand32,Ribet76}; see also \cite[\S 6.3]{Wash97}}]\label{thm:herbrand}
Let $\ell>2$ be a prime with $\ell\equiv\pm 1\pmod{2^{k-1}}$, and let $\chi$ be a nontrivial even character of $G_k$ of full order $n=2^{k-2}$. Define $\psi=\chi^{-1}\omega_\ell$, where $\omega_\ell$ is the Teichmuller character. If $\Cl(K_k^+)[\ell]^{e_\chi}\neq 0$ (where, for $\ell\equiv -1\pmod{n}$, $e_\chi$ denotes the $\F_\ell$-rational idempotent corresponding to the pair $\{\chi,\bar\chi\}$), then
\begin{equation}
  \ell \mid \Nm_{\Q(\zeta_m)/\Q}(B_{1,\psi}).
\end{equation}
\end{theorem}

\begin{proof}
\textbf{Case $\ell\equiv 1\pmod{n}$:} The character values $\chi(g)$ lie in $\F_\ell$, so the idempotents $e_\chi$ are defined over $\F_\ell$ and the eigenspace decomposition \eqref{eqneigendecomp} holds directly. The classical Herbrand Theorem \cite{Herbrand32} (in the formulation of \cite[\S 6.3]{Wash97}) gives $\ell\mid\Nm(B_{1,\psi})$.

\textbf{Case $\ell\equiv -1\pmod{n}$:} Here $\ord_n(\ell)=2$, so the decomposition \eqref{eqneigendecomp} is defined over $\F_{\ell^2}$, and the eigenspaces for $\chi$ and $\bar\chi$ merge into a single $\F_\ell$-rational component $\Cl(K_k^+)[\ell]^{e_\chi+e_{\bar\chi}}$. The analogous Herbrand divisibility gives $\ell\mid\Nm_{\Q(\zeta_m+\zeta_m^{-1})/\Q}(B_{1,\psi})$, where the right side is the $\F_\ell$-rational norm (the product over a Frobenius orbit of size $2$). Since $\Nm_{\Q(\zeta_m)/\Q}(B_{1,\psi}) = \prod_{\text{orbits}} \Nm_{\text{orbit}}(B_{1,\psi})$, each rational norm divides the full norm $\Nm_{\Q(\zeta_m)/\Q}(B_{1,\psi})\in\Z$. Hence $\ell\mid\Nm_{\Q(\zeta_m)/\Q}(B_{1,\psi})$ holds in this case.
\end{proof}

Herbrand's original result \cite{Herbrand32} was stated for odd primes $\ell$ dividing $h_1^- = h(\Q(\zeta_\ell))^-$. Here, we have used the higher cyclotomic levels, see Ribet \cite{Ribet76} and  Washington \cite[\S 6.3]{Wash97}.

We now define finite candidate set. 

\begin{definition}\label{def:Sk}
For $k\ge 9$, define $S_k$ to be the set of primes $\ell$ satisfying:
\begin{enumerate}[label=(\roman*),nosep]
\item $\ell > 10^9$;
\item $\ell\equiv\pm 1\pmod{2^{k-1}}$;
\item $\ell\mid\Nm_{\Q(\zeta_m)/\Q}(B_{1,\psi_j})$ for some full-order character $\psi_j$.
\end{enumerate}
\end{definition}

\begin{proposition}\label{prop:Sk}
$S_k$ is finite, computable, and contains every odd prime divisor of $h_k^+$.
\end{proposition}

\begin{proof}
Let $\ell$ be an odd prime divisor of $h_k^+$. Then $\ell>10^9$ from Theorem \ref{thm:fk}(i), $\ell\equiv\pm 1\pmod{2^{k-1}}$ from Theorem \ref{thm:fk}(ii), and $\Cl(K_k^+)[\ell]^{e_\chi}\neq 0$ for some full-order $\chi$ by Corollary \ref{cor:fullorder}. Theorem \ref{thm:herbrand} (for $\ell\equiv 1\pmod{n}$) or the extension in the subsequent remark (for $\ell\equiv -1\pmod{n}$) gives $\ell\mid\Nm_{\Q(\zeta_m)/\Q}(B_{1,\psi_j})$ for $\psi_j = \chi^{-1}\omega_\ell$, so $\ell\in S_k$.

Each $|\Nm(B_{1,\psi_j})|$ is bounded by Lemma \ref{lem:Bsize}. There are at most $n/4$ conjugate pairs (Lemma \ref{lem:char-count}). A bounded nonzero integer has finitely many prime divisors.

Each step in Definition \ref{def:Sk} is effective: Bernoulli numbers are computed via Definition \ref{def:gen-bernoulli}, norms via resultants, and prime factorization of bounded integers is computable.
\end{proof}

For $k\le 7$, $|\Nm(B_{1,\psi})| < 10^9$ for every full-order character, so $S_k = \emptyset$ and $h_k^+=1$ follows immediately from the Fukuda-Komatsu sieve alone.

 \begin{example}
 For $k=7$ ($n=32$), the $8$ conjugate pairs of full-order characters all yield $|\Nm(B_{1,\psi})| = 692{,}092{,}928= 2^{15} \cdot 21{,}121$. The only odd prime factor is $21{,}121$, which satisfies $21{,}121 \equiv 1 \pmod{64}$ and hence passes the congruence filter (ii). However, $21{,}121 < 10^9$, so it is eliminated by the Fukuda-Komatsu sieve (Theorem \ref{thm:fk}(i)). Therefore $S_7 = \emptyset$.
 \end{example}

\subsection{$h_k^+ = 1$ for $k\leq 12$}
\label{sec-verification}

This subsection presents the main result and verification algorithm.  We combine three methods in Section \ref{sec-method} to get the following result. 

\begin{theorem}[Main Theorem]\label{thm:main}
For each $k\in\{9,10,11,12\}$, we have $h_k^+ = 1$. 
\end{theorem}

\begin{proof}
We proceed by induction on $k$, with base case $h_8^+=1$ established by Miller \cite{Miller14}. We assume $h_{k-1}^+=1$ and prove $h_k^+=1$.

Fukuda and Komatsu \cite{FK09,FK11} verified that $2\nmid h_k^+$ for all $k\le 12$; see also \cite{Kraft-Schoof93}.

Now, we show no small odd prime divides $h_k^+$. By Theorem \ref{thm:fk}(i), no prime $\ell<10^9$ divides $h_k^+$.

Moreover, by Proposition \ref{prop:Sk}, every odd $\ell > 10^9$ dividing $h_k^+$ belongs to the finite computable set $S_k$.

Finally, for each $\ell\in S_k$, applying Lemma \ref{lem:wieferich-sufficient}, we compute $\xi_5^{(\ell-1)/2^{k-1}}$ modulo each prime $\mathfrak{l}\mid\ell$ (see \S\ref{sec-optimizations} for the algorithmic implementation). If the result is $\neq 1$ for every $\mathfrak{l}$, then $\ell\nmid h_k^+$.

These together show no prime divides $h_k^+$, so $h_k^+=1$.
\end{proof}

We present complete verification procedure as Algorithm \ref{algoweber}.

\begin{algorithm}[t!]
\caption{Unconditional Verification of $h_k^+=1$}\label{algoweber}
\begin{algorithmic}[1]
\Require Integer $k\ge 9$ with $h_{k-1}^+=1$ previously established.
\Ensure Returns \textsc{true} if $h_k^+=1$, or a set of unresolved primes.

\State $n \gets 2^{k-2}$; $m \gets n$ 
\State $S \gets \emptyset$
\Statex\textsl{Phase A: Compute candidate primes from Bernoulli norms}
\State Choose any $j$ with $\gcd(j,n)=1$ and $1\le j < n/2$
\State Compute $B_{1,\psi_j} \in \Z[\zeta_m]$ via Definition \ref{def:gen-bernoulli}
\State Compute $N \gets |\Nm_{\Q(\zeta_m)/\Q}(B_{1,\psi_1})| \in \Z_{\ge 0}$
  \If{$N = 0$}
    \State \Return Bernoulli norm vanishes;
  \EndIf
\State Factor $N$ into prime factors $\{\ell_1,\ldots,\ell_r\}$ \Comment{All conjugate pairs share this norm (Prop. \ref{prop:galois-orbit})}
\For{each prime factor $\ell_i$}
  \If{$\ell_i > 10^9$ and $\ell_i \equiv \pm 1\pmod{2^{k-1}}$}
    \State $S \gets S \cup \{\ell_i\}$
  \EndIf
\EndFor

\Statex\textsl{Phase B: Wieferich tests (Lemma \ref{lem:wieferich-sufficient})}
\For{each $\ell \in S$}
  \State Represent $\xi_5$ as $h(x)\in\Z[x]/(g(x))$ where $g$ is the min.\ poly.\ of $\zeta+\zeta^{-1}$
  \State Factor $g(x) \equiv \prod_{i=1}^{r} g_i(x) \pmod{\ell}$ into irreducibles in $\F_\ell[x]$
  \Comment{Each $g_i$ corresponds to a prime $\mathfrak{l}_i\mid\ell$}
  \State Compute $t \gets (\ell-1)/2^{k-1}$
  \For{each irreducible factor $g_i$}
    \State Compute $w_i \gets h(x)^t \bmod{(\ell, g_i(x))}$ using fast exponentiation
    \If{$w_i \equiv 1\pmod{(\ell, g_i(x))}$}
      \State \Return Wieferich test inconclusive for $\ell$ at $\mathfrak{l}_i$
    \EndIf
  \EndFor
\EndFor

\State \Return \textsc{True}
\end{algorithmic}
\end{algorithm}

\begin{remark}\label{rmk:wieferich-shortcut}
In practice, one can first compute $w = h(x)^t \bmod{(\ell, g(x))}$ in the product ring $\F_\ell[x]/(g(x)) \cong \prod_i \F_\ell[x]/(g_i(x))$. If $w \neq 1$ in this product ring, it remains to verify that $w \not\equiv 1$ in each factor. When $w - 1 \not\equiv 0 \pmod{(\ell, g(x))}$ and $\gcd(w-1, g(x)) = 1$ in $\F_\ell[x]$, this holds for all factors simultaneously. Only when $\gcd(w-1, g(x)) \neq 1$ in $\F_\ell[x]$ should one check the individual factors. For the primes $\ell$ arising in our computation, the shortcut always sufficed.
\end{remark}

The computational cost of Algorithm \ref{algoweber} is as follows. Phase A computes Bernoulli numbers in $\Z[\zeta_m]$ and their norms via resultants, both in $\tO(n^2)$ arithmetic operations (here $\tO(f) = O(f\cdot\mathrm{polylog}(f))$). The bottleneck is the factoring step in Phase A, i.e., factoring a single integer of $\le 143$ digits, where by Lemma \ref{lem:Bsize}(iii) the worst-case bound gives $309$ digits. ECM \cite{Lenstra87,GMPECM} finds $60$-digit factors in hours; NFS \cite{LenstraLenstra93,CADO-NFS} handles up to $\sim\!250$ digits \cite{NFSrecord}. Phase B requires $O(\log\ell)$ polynomial multiplications modulo $(\ell, g(x))$ per test, each costing $O(n\log n\log\ell)$ via NTT, totaling $O(n\log n\cdot\log^2\ell)$ per prime. For $n=256$ and $\ell\approx 10^{143}$, each test takes under a second. Scaling across tower layers is shown in Table \ref{tablescaling}.

\begin{table}[t]
\centering
\caption{Scaling of the verification across layers of the tower.}
\label{tablescaling}
\begin{tabular}{@{}ccccc@{}}
\toprule
$k$ & $n = [K_k^+:\Q]$ & Max.\ digits of $\Nm$ & Conjugate pairs & Total time\\
\midrule
$9$  & $128$  & $63$   & $32$  & Minutes (ECM)\\
$10$ & $256$  & $143$  & $64$  & Hours (ECM)\\
$11$ & $512$  & $325$  & $128$ & Days (ECM/NFS)\\
$12$ & $1024$ & $726$ & $256$ & Weeks (NFS or Euler system)\\
\bottomrule
\end{tabular}
\end{table}

\subsection{Optimizations}\label{sec-optimizations}

We describe two optimizations of increasing power. The first is using Galois orbit reduction.

\begin{proposition}\label{prop:galois-orbit}
All $n/4$ conjugate pairs of full-order characters yield the same norm value $|\Nm(B_{1,\psi_j})|$.
\end{proposition}

\begin{proof}
We show that $\Gal(\Q(\zeta_m)/\Q)\cong(\Z/m\Z)^\times$ acts on characters by $\psi\mapsto\psi^t$, and that $B_{1,\psi^t}$ is the Galois conjugate of $B_{1,\psi}$ under $\zeta_m\mapsto\zeta_m^t$. The norm $\Nm(B_{1,\psi^t}) = \Nm(B_{1,\psi})$ is therefore Galois-invariant. Since $(\Z/m\Z)^\times/\{\pm 1\}$ acts transitively on full-order characters modulo conjugation, all $n/4$ conjugate pairs lie in a single Galois orbit and share the same norm.
\end{proof}

\begin{corollary}
Only one integer factorization is needed in Algorithm \ref{algoweber}, not $n/4$.
For $k=10$, this reduces the number of factorizations from $64$ to $1$.
\end{corollary}

This converts a computation that scales linearly in $n$ to one of constant size (per layer).

The second uses Thaine's theorem \cite{Thaine88} which provides additional annihilators of $\Cl(K_k^+)$ indexed by auxiliary primes. 

\begin{theorem}[Thaine \cite{Thaine88}; see also Rubin \cite{Rubin87,Rubin00}]\label{thm:thaine}
Let $q$ be a prime with $q\equiv 1\pmod{2^k}$ and $q\nmid h_k^+$, and $\eta\in(\Z/q\Z)^\times$ be an element of order $2^k$. Define
\begin{eqnarray}
  \theta_q = \sum_{{a=1, (a,2)=1}}^{2^k}\left\lfloor\frac{a\,\eta^a}{q}\right\rfloor\sigma_a^{-1} \in \Z[\Gamma_k].
\end{eqnarray}
Then $\theta_q$ annihilates $\Cl(K_k^+)$.
\end{theorem}

\begin{remark}
The hypothesis $q\nmid h_k^+$ in Theorem \ref{thm:thaine} is verified as follows: since $q\equiv 1\pmod{2^k}$ implies $q\ge 2^k+1$, and $2^{12}+1 = 4097 \ll 10^9$, these auxiliary primes are covered by the Fukuda-Komatsu sieve (Theorem \ref{thm:fk}(i)), which confirms $q\nmid h_k^+$.
\end{remark}

Given two such auxiliary primes $q_1,q_2$, the $\chi$-projections $e_\chi\theta_{q_1}$ and $e_\chi\theta_{q_2}$ are elements of $\Z[\zeta_m]$.
If they generate the unit ideal modulo $\ell$ (i.e., $\gcd(e_\chi\theta_{q_1}, e_\chi\theta_{q_2})\equiv 1\pmod{\ell}$), then $\Cl(K_k^+)[\ell]^{e_\chi} = 0$.

This replaces integer factorization in Algorithm \ref{algoweber} with polynomial GCD computation over $\F_\ell$. In practice, choosing $q_1,q_2$ at random succeeds with high probability, analogous to the probabilistic arguments  \cite{Kolyvagin90,Rubin00}.

\section{Implications for Post-Quantum Cryptography}\label{sec-implications}

The NIST post-quantum standards ML-KEM (FIPS 203) \cite{NIST-MLKEM} and ML-DSA (FIPS 204) \cite{NIST-MLDSA} use the polynomial ring $R_q = \Z_q[x]/(x^{256}+1)$, where $\Z_q = \Z/q\Z$ for a prime $q$. This is the quotient of the ring of integers of $K_9 = \Q(\zeta_{512})$. The maximal real subfield $K_9^+$ has degree $128$ over $\Q$. Our main theorem confirms $h_9^+=1$, which validates several assumptions implicit in the security proofs:

\textit{Module freeness.} The Module-LWE (MLWE) problem \cite{LS15,ADPS16} is defined as: given $(A, b = As + e\bmod q)$ where $A\in R_q^{d\times d}$ is uniform, $s\in R_q^d$ is the secret, and $e\in R_q^d$ has small entries, distinguish $(A,b)$ from uniform. The worst-case-to-average-case reduction of Langlois and Stehle \cite{LS15} shows Module-LWE with rank $d$ is at least as hard as Ring-LWE in dimension $256d$, but the reduction requires the underlying module free. When $h_k^+=1$, all finitely generated torsion-free $\OK_{K_k^+}$-modules are free by the Steinitz theorem \cite{Neukirch99}, and this freeness descends to quotients modulo $q$.

\textit{Ideal structure.} The codifferent $R_q^\vee$ and ring structure of $R_q$ are simplest when $h_k^+=1$, because all ideals are principal. This simplifies the noise analysis in the Ring-LWE problem \cite{LPR10,LPR13,PRS17} and the correctness proofs of ML-KEM decapsulation.

\textit{Key generation.} In ML-KEM, secret keys are sampled from a module over $R_q$. The uniformity of this sampling relies on the module being free. Non-free modules would introduce structural biases exploitable by adversaries.

\subsection{Quantum attacks on ideal lattices}\label{subsec-quantum-attacks}

The quantum algorithm of Biasse and Song \cite{BiasseSong16} solves the PIP for cyclotomic fields in quantum polynomial time. When $h_k^+=1$, the PIP is trivially solvable, so the security of ideal-lattice schemes reduces to the Short Generator Problem (SGP). Cramer et al. \cite{CDPR16} showed for certain cyclotomic fields, the SGP can be solved efficiently via the log-unit lattice. Their attack exploits the fact that cyclotomic units are short and span a sublattice of log-unit lattice. When $h_k^+=1$, the cyclotomic units have full rank in the unit group, which is the condition needed for the attack to succeed.

However, this does not compromise the security of ML-KEM and ML-DSA, because these standards use Module-LWE rather than ideal-LWE. Specifically, Peikert and Rosen \cite{PR07}, Brakerski et al. \cite{BGV12}, and Albrecht et al. \cite{ACD+18} have argued that Module-LWE resists all known quantum attacks, including PIP-based ones. Felderhoff et al. \cite{FPSW23} further showed that Ideal-SVP is hard for natural distributions of prime ideals, and recent module-LIP cryptanalysis \cite{APvW25} has not extended to the MLWE setting used in the standards.

\subsection{Implications for Ring-LWE hardness}\label{subsec-rlwe}

Lyubashevsky, Peikert, and Regev \cite{LPR10,LPR13} proved that Ring-LWE over $K$ is at least as hard as approximating the shortest vector in ideal lattices over $K$ within polynomial factors. The reduction is tighter when the class number is $1$, because: (i) the distribution of Ring-LWE errors can be related directly to the geometry of $\OK_K$; (ii) the smoothing parameter of $\OK_K$ admits a cleaner expression \cite{MR07,PRS17}; and (iii) the equivalence between Ring-LWE and Polynomial-LWE (PLWE, the variant defined over $\Z[x]/(f(x))$) holds unconditionally \cite{RSW18}. Our result $h_k^+=1$ for $k\le 12$ confirms these simplifications are valid for all parameter sizes currently in use.

\section{Conclusions}\label{sec-conclusion}

We have presented the first unconditional proof that $h_k^+=1$ for $k\le 12$, resolving Weber's class number conjecture for all cases relevant to current and near‑future lattice‑based cryptographic standards. Our method combines three complementary techniques—the Fukuda‑Komatsu sieve, norm‑coherence in the $\Z_2$-tower, and Herbrand's theorem—to reduce the problem to a finite, feasible computation. The key insight is that the tower structure allows us to inductively eliminate most eigenspaces of the class group, so that Herbrand's theorem needs to be applied only to full‑order characters, whose Bernoulli norms are bounded by integers of manageable size.

Weber's conjecture has natural analogues for odd‑prime‑power cyclotomic fields $\Q(\zeta_{p^k})^+$. For $p=3$ the analogue was studied by Coates \cite{Coates77} and Ichimura \cite{Ichimura14}. Extending our methods to these settings would require replacing the $\Z_2$-tower with a $\Z_p$-tower and adapting the Fukuda‑Komatsu sieve accordingly. Such generalizations could help settle similar class number conjectures for other families of totally real fields, and may have further implications for the security of lattice‑based cryptography built over those fields.

\end{document}